\documentclass[times,11pt,doublespace]{ettauth}

\usepackage{moreverb}
\usepackage{amsthm}
\usepackage{graphicx}
\usepackage{balance}
 \usepackage{stfloats}

\usepackage[colorlinks,bookmarksopen,bookmarksnumbered,citecolor=red,urlcolor=red]{hyperref}
\newtheorem{lemma}{Lemma}
\newtheorem{proposition}{Proposition}

\newcommand\BibTeX{{\rmfamily B\kern-.05em \textsc{i\kern-.025em b}\kern-.08em
T\kern-.1667em\lower.7ex\hbox{E}\kern-.125emX}}

\def\volumeyear{2013}

\begin{document}


\runningheads{Z. Fei et al.}{}

\articletype{RESEARCH ARTICLE}

\title{Improved LT Codes in Low Overhead Regions for Binary Erasure Channels}

\author{Zesong Fei\affil{1}\corrauth\, Congzhe Cao\affil{1},
Ming Xiao \affil{2}, Iqbal Hussain \affil{2} and  Jingming Kuang\affil{1}}

\address{\affilnum{1}School of Information and Electronics, Beijing Institute of Technology, Beijing 100081, China\\
\affilnum{2}School of Electrical Engineering, The Royal Institute of Technolgoy, Stockholm 10044, Sweden}

\corraddr{Zesong Fei is with School
of Information and Electronics, Beijing Institute of Technology, Beijing,
China. E-mail: feizesong@bit.edu.cn }

\begin{abstract}
We study improved degree distribution for Luby Transform (LT) codes which exhibits improved bit error rate performance particularly in low overhead regions. We construct the degree distribution by modifying Robust Soliton distribution. The performance of our proposed LT codes is evaluated and compared to the conventional LT codes via And-Or tree analysis. Then we propose a transmission scheme based on the proposed degree distribution to improve its frame error rate in full recovery regions. Furthermore, the improved degree distribution is applied to distributed multi-source relay networks and unequal error protection. It is shown that our schemes achieve better performance and reduced complexity especially in low overhead regions, compared with conventional schemes.
\end{abstract}


\maketitle


\section{Introduction}
Digital fountain codes have been introduced in [1] for binary erasure channels (BECs) with unknown and time varying erasure probabilities. In contrast to conventional fixed-rate codes like Low density parity check(LDPC) codes [2] [3], fountain codes are rateless codes capable of providing potentially an unlimited number of encoded symbols from a limited block of source symbols. Luby Transform (LT) codes [4] were the first practical rateless codes and Raptor codes [5] were later proposed to provide better performance by precoding. It has been shown in [4] that LT codes based on Robust Soliton distribution (RSD) exhibit good performance over BECs with unknown erasure rates, provided with sufficient overhead. However, a large bit error rate (BER) has been observed for LT codes in low overhead regions [6], [7]. A multi-source relay scheme based on distributed LT (DLT) codes has been studied in [8] where a new degree distribution is formed at the relay node. However, coding complexity is relatively high due to the spike existing in RSD. To address this problem, the scheme in [9] proposed a soliton like distribution where high degrees are spread across many degrees instead of being concentrated at a single degree. Yet the scheme requires a large buffer at the relay node. The equivalence of performance of distributed LT codes and LT codes with related parameters in the asymptotic regime is shown in \cite{AndOrTreeDLT}. In \cite{DecomposedLT}, the authors studied decomposed LT codes comprising of two layers of encoding performed collaboratively by the source and relay nodes to reduce transmission latency and energy consumption.

In this paper, to reduce the BER of LT codes in low overhead regions, we modify the degree distribution for LT codes. Meanwhile,  we propose a point-to-point transmission scheme based on the proposed degree distribution to improve the frame error rate (FER) in full recovery regions. Then we apply our proposed degree distribution to DLT codes. The performance improvement and complexity reduction have been achieved as compared to conventional LT (and DLT) codes.  Moreover, we extend our proposed scheme to achieve unequal error protection (UEP), which cannot be achieved by conventional DLT schemes in [8].

\section{Background Information}
\subsection{A Review of RSD}

RSD is a widely used degree distribution for fountain codes, which shows good performance with sufficient overhead [4]. It is defined as \(u(i)\) in the following:

The Ideal Soliton distribution is defined as \(\rho(1),...,\rho(k)\)
\begin{equation}
\left\{ \begin{array}{l}
 \rho (1) = 1/k, \\
 \rho (i) = 1/i(i - 1), i = 2,...,k, \\
 \end{array} \right.
 \end{equation}
where \(k\) is the number of source symbols.  Then, \(\Gamma(i)\) is defined as:
\begin{equation}
\Gamma (i) = \left\{ \begin{array}{l}
 R/ik,\;\;\;\;\;\;\;\;\;\;\;\;\;\;\;i = 1,...,\frac{k}{R} - 1,\;\;\;\;\;\;\;\;\;\;\;\; \\
 R\ln (R/\delta )/k,\;\;\;i = \frac{k}{R}\;, \\
 0,\;\;\;\;\;\;\;\;\;\;\;\;\;\;\;\;\;\;\;\;\;i = \frac{k}{R} + 1,...,k, \\
 \end{array} \right.
\end{equation}
where
\begin{equation}
 R = c\ln (k/\delta )\sqrt k.
\end{equation}

The parameters c and $\delta$ have significant impact on the performance of LT codes. RSD is obtained by adding the Ideal Soliton distribution \(\rho(i)\)  to \(\Gamma(i)\) and normalizing the result:
\begin{equation}
u(i) = (\rho (i) + \Gamma (i))/\beta ,i = 1,...,k,
\end{equation}
where \(\beta  = \sum\limits_{i = 1}^k {(\rho (i) + \Gamma (i)} )\).

We denote $n$ as the number of encoded symbols and hence the overhead $\gamma $ is defined as \(\gamma  = (n - k)/k\). It is shown in [4] that when $n = k\beta$, the frame error rate is at most $\delta$. Therefore, LT codes based on RSD provide low frame error rates in full recovery regions with medium to large overhead (i.e. $n \ge k\beta$). Yet in low overhead regions (i.e. $k\beta \ge n \ge k$), high BER is observed for LT codes based on RSD. We note that the partial recovery is also very important in some applications such as multimedia content delivery, as mentioned in [7] where the authors focus on intermediate recovery, namely $n \le k$, and only part of source symbols are recovered. Although we also aim to achieve low bit error rate for partial recovery, however, we emphasize that in this paper we focus on \emph{low overhead regions}, where all source symbols will be recovered and $n$ is slightly larger than $k$.

We note that the considered scenario is also suitable for applications with multiple users. For example, the multimedia broadcast/multicast services (MBMS) has a time-limited broadcast delivery phase [12]. As the delivery phase is time-limited in broadcast scenario, the user equipments (UEs) can only collect a limited number of encoded symbols and thus they are in low overhead regions.  Therefore, the optimized code design for low overhead regions is an interesting problem.

\subsection{System Model}

We first introduce the improved LT codes and the scheme for improving FER. Then the codes are extended to multi-source relay scenarios where two source nodes  intend to send their data to the same destination  via a common relay node. At each source node, $k/2$ source symbols are encoded by an LT encoder (according to degree distributions to be discussed later) to generate sequences $E_{1}$ and $E_{2}$ of length $n$ respectively, which are transmitted to the relay node through  error-free channels. At the relay node, one of the $i$-th $(i = 1, 2, \cdots ,n)$ encoded symbols (denoted by $E_{1,i}$ and $E_{2,i}$ for $E_{1}$ and $E_{2}$ respectively) is directly forwarded with probability $\lambda$ (while another is discarded) or they are XORed and forwarded with probability $(1 - \lambda)$. At the receiver node,  iterative decoding is performed to recover $k$ source symbols after sufficient encoded symbols are received.

As mentioned above, RSD is originally proposed to achieve a low FER. Yet, it does not address the problem of minimizing BERs especially in low overhead regions. To solve this problem, we shall modify RSD. Clearly in RSD,  the ripple size is the number of input symbols covered by the degree-1 encoded symbols in the decoding process [4]. The ratio (size) of the ripple is critical for the design of LT codes. If the size of the ripple is too small, then the decoding failure may occur. The spike (a high degree generated from the degree distribution with relatively high probability) is another important parameter of RSD which ensures that all the source symbols are encoded. However, the spike may not be useful (even disadvantage) for LT codes in low overhead regions. It is because that in order to release the high-degree encoded symbols produced by the spike, a large number of source symbols connecting to them need to be released, which in turn requires more encoded symbols received. Thus, it is very unlikely to recover them in low overhead regions. Here we design an improved degree distribution of LT codes which achieves better BERs in low overhead regions.

\section{And-Or tree Analysis}

We first derive an improved degree distribution \(u_i(i)\) for low overhead. The basic idea is to reduce the ratio of spike and appropriately increase the initial size of the ripple. And-Or tree lemma as follows will be utilized in our analysis.

\begin{lemma} (And-Or tree lemma) The probability \(y_\infty = \mathop {\lim }\limits_{l \to \infty } {y_l}\) that a source bit is not recovered is given asymptotically as [10]:
\begin{equation}
\left\{ \begin{array}{l}
 {y_0} = 1, \\
 {y_l} =  \exp ( - (1 + \gamma )\Omega '(1 - {y_{l - 1}})), \\
 \end{array} \right.
\end{equation}
where $l$ is the number of iterations, \(\Omega (x)\) is the check node degree distribution with constant average degree and \(\Omega '(x)\) is its derivative with respect to $x$. Resulting from the And-Or tree lemma, \(y_{l}\) converges to a fixed point asymptotically, which is the BER.
\end{lemma}

With the lemma, we can design an improved degree distribution based on a target $k$ for LT codes with low overhead. Denote \(\sigma  \in (0,1]\) as the ratio of the spike being reduced from RSD (degree-$1$ node is increased accordingly). Then we can obtain the equations as below to analyze its asymptotic behavior:
\begin{equation}
\left\{ \begin{array}{l}
 {\Omega _{1'}}({x_1}) = {\Omega _1} + \sigma {\Omega _{\frac{k}{R}}} + 2{\Omega _2}{x_1} + 3{\Omega _3}{x_1}^2 \\
 \;\;\;\;\;\;\;\;\;\;\;\;\; + .... + \frac{k}{R}(1 - \sigma ){\Omega _{\frac{k}{R}}}{x_1}^{\frac{k}{R} - 1} + ...., \\
 {\Omega _{2'}}({x_2}) = {\Omega _1} + 2{\Omega _2}{x_2} + 3{\Omega _3}{x_2}^2 + ....\; +  \\
 \;\;\;\;\;\;\;\;\;\;\;\;\;\;\frac{k}{R}{\Omega _{\frac{k}{R}}}{x_2}^{\frac{k}{R} - 1} + ...., \\
 \end{array} \right.
\end{equation}
where \(\Omega_1 ({x_{1}})\) is our designed degree distribution, and \(\Omega_2 ({x_{2}})\) is RSD with the target $k$. Note \({\Omega _{\frac{k}{R}}}\) is the probability that the spike is selected in RSD. In the following, we investigate the iterative process shown by And-Or tree lemma. We replace \(x_{1}\) and \(x_{2}\) by \(x_{t,1}\) and \(x_{t,2}\) to denote their values in the $t$-th iteration, where \(x_{t,1}=1-y_{t-1,1}\) and $x_{t,2}=1-y_{t-1,2}, t \in [1,\infty )$. Firstly the main result of this section is presented as follows.

\begin{proposition}
For LT codes with low overhead where the BER of RSD is larger than a threshold (to be discussed), the degree distribution in which $\sigma, (0 < \sigma \leq 1)$ ratio of spike is reduced (degree-$1$ nodes are increased accordingly) outperforms RSD in BER performance. Moreover, the BER is minimized when $\sigma = 1$.
\end{proposition}

\begin{proof}
Based on (1) and (2), we need to find ${\Omega _{1}}'({x_{\infty ,1}}) > {\Omega _{2}}'({x_{\infty ,2}})$ such that our proposed degree distribution has lower BER than RSD for a given \(\gamma \). To begin with, we consider the first iteration, \(\Omega_1 '({x_{1,1}}) = \Omega {}_1 + \sigma {\Omega _{\frac{k}{R}}} > \Omega_2 '({x_{1,2}}) = \Omega {}_1\), given \(y_{0,1}=y_{0,2}=1\). Thus, \({y_{1,1}}\) is smaller than \({y_{1,2}}\). For the subsequent iterations, we must have \(\Omega_1 '({x_{t,1}}) > \Omega_2 '({x_{t,2}})\) to ensure \({y_{t,1}} < {y_{t,2}}\) until the final iteration, achieving \({y_{\infty,1}} < {y_{\infty,2}}\) . Thus, for any iteration $t$ we derive the following inequality which should be satisfied:
\begin{equation}
\begin{array}{l}
 \sigma {\Omega _{\frac{k}{R}}} + \frac{k}{R}(1 - \sigma ){\Omega _{\frac{k}{R}}}{x_{t,1}}^{\frac{k}{R} - 1} - \frac{k}{R}{\Omega _{\frac{k}{R}}}{x_{t,2}}^{\frac{k}{R} - 1} > 0. \\
 \end{array}
\end{equation}

Consequently, to meet (3), we have \(x_{t,1} > x_{t,2}\) in each iteration. More strictly, by substituting $x_{t,2}$ for $x_{t,1}$, the following inequality must be held:
\begin{equation}
\begin{array}{l}
 \sigma {\Omega _{\frac{k}{R}}} + \frac{k}{R}(1 - \sigma ){\Omega _{\frac{k}{R}}}{x_{t,2}}^{\frac{k}{R} - 1} - \frac{k}{R}{\Omega _{\frac{k}{R}}}{x_{t,2}}^{\frac{k}{R} - 1} > 0 \\
  \Rightarrow 1 > \frac{k}{R}{x_{t,2}}^{\frac{k}{R} - 1}. \\
 \end{array}
\end{equation}

During the iterative process, \(x_{t,2}\) keeps increasing.  Thus \(x_{\infty,2}\) is the maximum value. Therefore, we just need to consider the final iteration, where \(f(x_{\infty,2},k)= \frac{k}{R}{x_{\infty,2}}^{\frac{k}{R} - 1}<1 \) should be satisfied to achieve \(\Omega_1 '({x_{\infty,1}})>\Omega_2 '({x_{\infty,2}})\). For a fixed $k$, $f(x_{\infty,2},k)$ increases with larger $x_{\infty,2}$. Thus, for a given $k$,  $x_{\infty,2}$ should be smaller than a threshold. In other words, \(y_{\infty,2}\) should be larger than a certain threshold so that inequality (4) can be satisfied, observing \(x_{\infty,2}=1-y_{\infty,2}\). The BER of LT codes is also decided by $\gamma$. Therefore, if the BER of RSD \(y_{\infty,2}\) is relatively large (larger than a threshold), namely the overhead is small so that $1 > f(x_{\infty,2},k)$ is satisfied then our proposed degree distribution has better BER performance.


We note that in inequality (4) (below part), parameter \(\sigma \) is removed. Thus, the proposed distribution always outperforms RSD in BER performance, if (4) is satisfied. Here we use the optimization method to find the optimal ratio of removed spike when the minimal BER is achieved. Following (4), the optimization problem (with respect to \(\sigma\)) is formulated as follows:
\begin{align}
\begin{array}{l}
 \mathop {\max }\limits_{\sigma} \;\;\;S = \sigma {\Omega _{\frac{k}{R}}} +  \frac{k}{R}(1 - \sigma ){\Omega _{\frac{k}{R}}}{x_{t,2}}^{\frac{k}{R} - 1} \\ \;\;\;\;\;\;\;\;\;\;\;\;\;\;\;\;\;\;- \frac{k}{R}{\Omega _{\frac{k}{R}}}{x_{t,2}}^{\frac{k}{R} - 1}  \\
 \;\;s.t.\;\;\;\;\sigma  \in (0,1],
 \;\;\;\;\frac{k}{R}{x_{t,2}}^{\frac{k}{R} - 1} < 1 \\
 \end{array}.
\end{align}

It is easy to obtain \(\sigma  = 1\), which is a simple but interesting result. Thus, asymptotically, in low overhead regions where the BER of RSD is larger than the threshold, the proposed degree distribution outperforms RSD in BER performance for a fixed $k$. This concludes the proof. We further note that the conclusion also holds non-asymptotically (for small or medium $k$), as it will be shown in our numerical results later.
\end{proof}

Now, the proposed degree distribution $u_i(i)$ can be derived as below. $\rho(i)$ and $\Gamma(i)$ are updated as $\rho_i(i)$ and $\Gamma_i(i)$ as follows
\begin{equation}
\left\{ \begin{array}{l}
 \rho_i (1) = 1/k + R\ln (R/\delta )/k,\\
 \rho_i (i) = 1/i(i - 1),i = 2,...,k \\
 \end{array} \right.
 \end{equation}

\begin{equation}
\Gamma_i (i) = \left\{ \begin{array}{l}
 R/ik,i = 1,...,\frac{k}{R} - 1 \\
 0,i = \frac{k}{R},...,k \\
 \end{array} \right.
\end{equation}

Then $u_i(i)$ is derived as

\begin{equation}
u_i(i) = (\rho_i (i) + \Gamma_i (i))/\beta ,i = 1,...,k.
\end{equation}

\section{Discussions}
\subsection{BER Lower Bound}

For the threshold in Proposition 1, we have the following BER lower bound $\widehat{y}_{\infty,2}$ such that if the BER of RSD \(y_{\infty,2}\) is larger than the lower bound due to small overhead then $1 > f(x_{\infty,2},k) $ is always satisfied to ensure lower BER of our proposed degree distribution compared to RSD [4]. The BER lower bound \(\widehat{y}_{\infty,2}\) can be computed as the critical value of inequality (8)

\begin{equation}
{{\hat y}_{\infty ,2}} = 1 - {(\frac{k}{R})^{(1/(1 - \frac{k}{R}))}}
\end{equation}
and is depicted in Fig. 1. We note that it is the typical property of RSD that the BER is at a very high level thus will exceed the lower bound \(\widehat{y}_{\infty,2}\) when the overhead is low [6], [7].

\begin{figure}[!ht]
\includegraphics[scale=0.37]{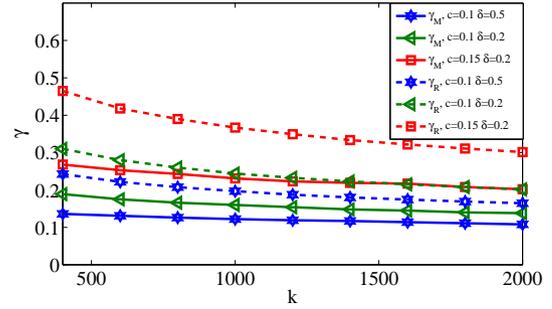}
\caption{ $\gamma_M
$ and $\gamma_R$ versus $k$} \label{system.eps}
\end{figure}
\begin{figure}[!ht] \centering
\includegraphics[scale=0.37]{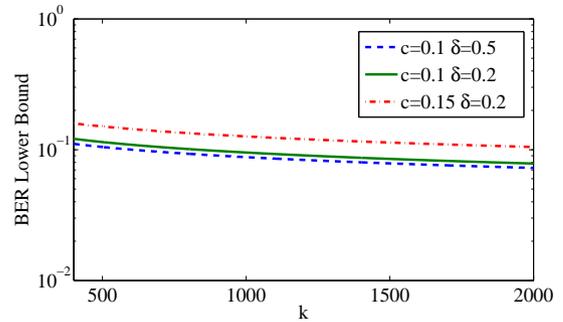}
\caption{ BER Lower bound  \(\widehat{y}_{\infty,2}\)} \label{system.eps}
\end{figure}

\setcounter{figure}{3}
\begin{figure*}[b]
\centering
\includegraphics[scale=0.85]{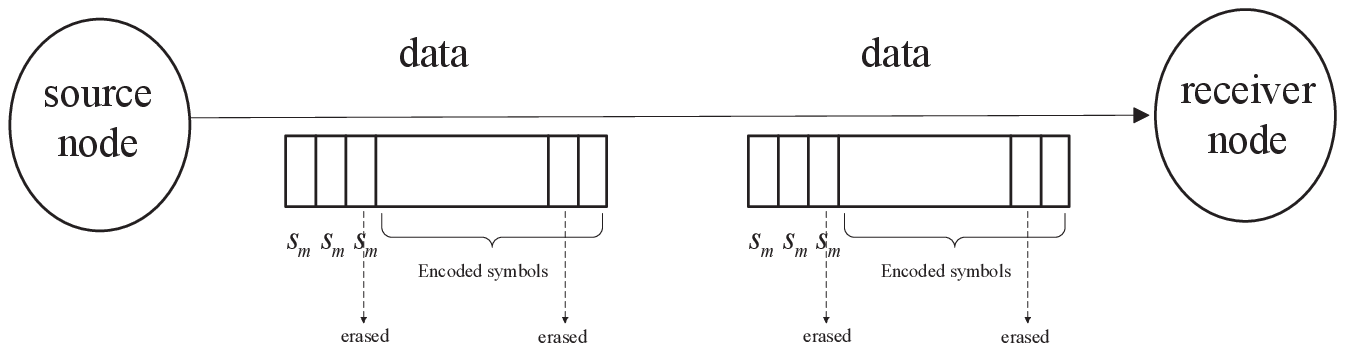}
\caption{The proposed transmission scheme for improving FER}
\end{figure*}

\subsection{Overhead}

The overhead region where the proposed degree distribution has better BER performance than RSD may be a concern. To clarify, we define two parameters $\gamma_R$ and $\gamma_M$, $\gamma_R$ denotes the overhead needed to ensure low frame error rate when RSD is utilized. As mentioned above, $\gamma_R$ is defined as
\begin{equation}
\gamma_R = (k\beta - k)/k
\end{equation}

As to $\gamma_M$, a set $\alpha  = \left\{ {{\gamma _i}} \right\}$ is introduced to describe the overhead region where the proposed degree distribution is ensured to outperform RSD in BER performance. And then $\gamma_M$ is defined as ${\gamma _M} = \max \{ {\gamma _i}\} $. Based on Proposition 1 and the And-Or tree lemma, $\alpha$ is derived according to $\widehat{y}_{\infty,2}$, which means the following conditions
\begin{equation}
\left\{ \begin{array}{l}
 {y_0} = 1, \\
 {{\hat y}_{\infty ,2}} \le \mathop {\lim }\limits_{l \to \infty } \left\{ {{y_l} = \exp ( - (1 + {\gamma _i})\Omega '(1 - {y_{l - 1}}))} \right\} \\
 \end{array} \right.
\end{equation}
should be satisfied for $\forall i$. Then, $\gamma_M$ can be obtained.

Fig. 2 depicts $\gamma_M$ and $\gamma_R$ according to different $k$. It can be seen that $\gamma_M$ is smaller than $\gamma_R$, which means in low overhead regions the proposed degree distribution is ensured to outperform RSD. In fact, the performance gain can still be achieved even when $\gamma > \gamma_M$, which will be shown later in the simulation results.

\subsection{complexity comparison}

In what follows, we give analytical results to compare encoding complexity for RSD and the proposed degree distribution respectively. The complexity is measured by average degree of the encoded symbols. Based on the structure of the proposed degree distribution, it is clear that lower complexity is achieved as the spike is removed and the ratio of degree-one nodes increases. Since \(\Omega '(1)\) is the the average degree of an encoded symbol, we can explicitly evaluate the reduction of encoding complexity \(\Delta \) as
\begin{equation}
\Delta  = {u}'(1) - u_i'(1) = \frac{k}{R}{\Omega _{\frac{k}{R}}} - {\Omega _{\frac{k}{R}}},
\end{equation}
where ${u}'(i)$ and \(u_i'(i)\) are derivatives of RSD and \(u_i(i)\) with the same $k$.

In Fig. 3, we show the complexity comparison of RSD and the improved degree distribution. It can be seen from the figure that lower complexity is achieved by the improved degree distribution.
\setcounter{figure}{2}
\begin{figure}[!ht] \centering
\includegraphics[scale=0.37]{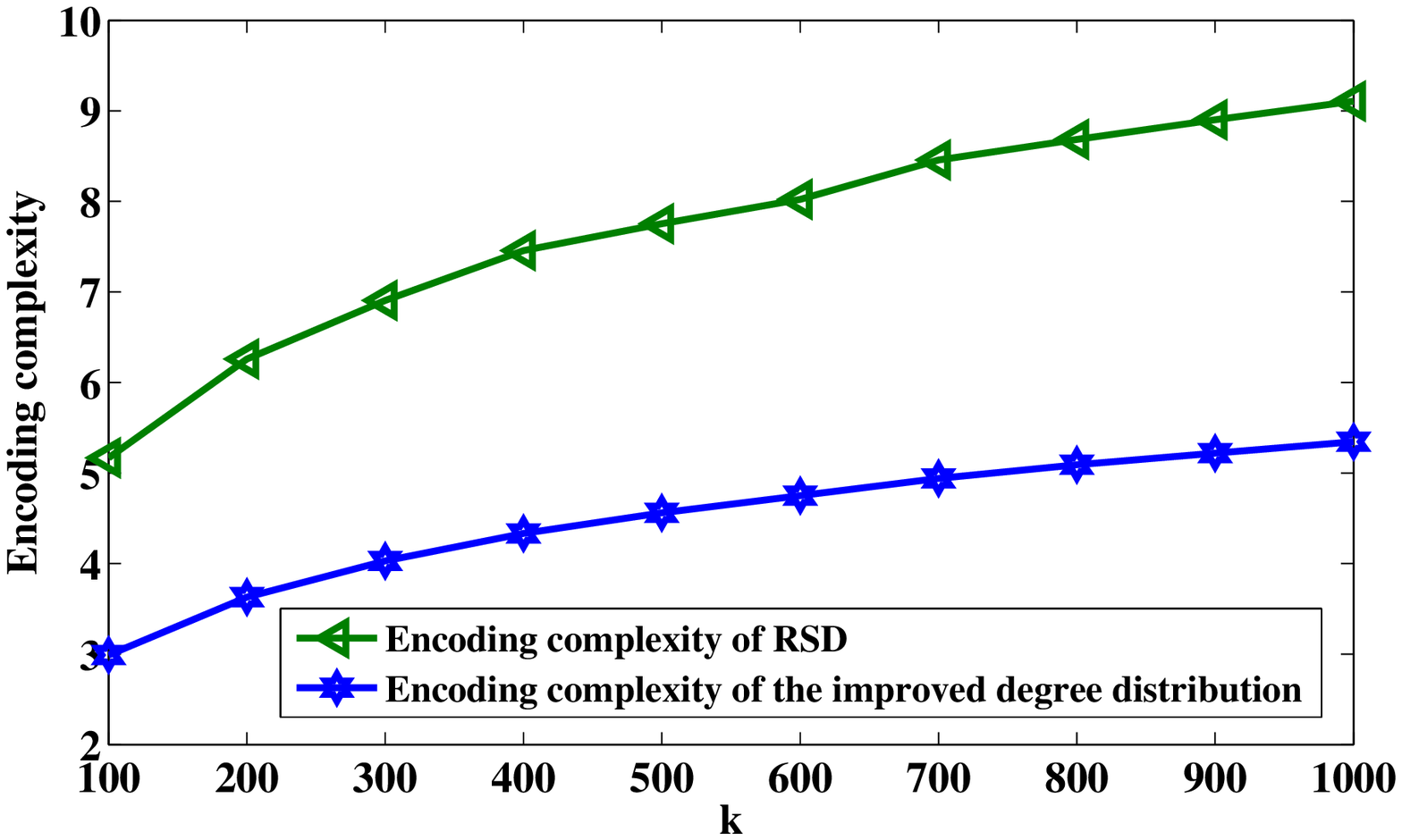}
\caption{Encoding complexity comparison, $c=0.15, \delta=0.2$} \label{system.eps}
\end{figure}

Finally, we would like to mention that we could reduce spike and increase degree-$2$ or $3$ nodes and so on. However, this will lead to a degrading performance of decoding probability as well as higher encoding and decoding complexity. Therefore, we only consider the scenario where the spike is removed and degree-$1$ nodes are increased, namely increasing the size of ripple  and reducing the size of spike.

\section{Improved FER for the improved degree distribution}

Although we focus on the code design in low overhead regions, it is straightforward that our codes can be extended to medium to large overhead (full recovery regions). Yet degraded FER performance in full recovery regions is observed compared with RSD due to reduced spike. In order to compensate for the degraded FERs, we propose a transmission scheme that has better BER performance in low overhead regions than RSD while trying to improve its FER performance in full recovery regions.

The degraded FER is caused by lower variable and check node degrees in the proposed degree distribution. Therefore, there might be some source symbols that are not covered by encoded symbols in full recovery regions. Considering the generator matrix of the LT encoder with a fixed $k$, the number of source symbols that are not covered by encoded symbols $m$ can be approximately determined as
\begin{equation}
m = k{(1 - \frac{{u_i'(1)}}{k})^n}.
\end{equation}

In full recovery regions, $m$ is smaller than one with $k$ in the thousands (when $n=k\beta$), which means mostly only one source symbol is not covered. For example, it can be verified that $m=0.68$ when $k=500$, $m=0.65$ when $k=2000$, and $m=0.618$ when $k=5000$, all for $c=0.15, \delta=0.2$. Based on above observations, we propose a transmission scheme for the proposed degree distribution to recover the source symbol, in order to improve the FER performance in full recovery region. The scheme is shown in Fig. 4. At the source node, after the entire encoded symbols are generated by the LT encoder, a special encoded symbol $s_m$ that is the sum (in binary) of all $k$ source symbols is generated and then all the encoded symbols are transmitted to the destination node. Assume that in the full recovery region, the decoder gets stalled when $k-1$ source symbols have been recovered. Then we obtain the sum of all $k-1$ recovered symbols $s_{r}$ and compute the bitwise XOR of $s_{r}$ and $s_m$. Thus, the source symbol that has not been recovered previously can be reconstructed. Clearly the FER is greatly impacted by the reception of $s_m$, which may be also lost in erasure channels.  Thus, we may need to send multiple $s_m$. Then the probability of not receiving $s_m$ is negligible.

\section{Application to DLT and unequal error protection}

DLT codes were first introduced in [8] by using a deconvolution method. Yet a direct deconvolution of RSD does not necessarily yield a valid probability distribution. As discussed in [8], if we try to reproduce $u(i), i>1 $ by recursively solving $\rho(i)$ from direct deconvolution 
\begin{equation}\label{deconvolution}
\begin{array}{l}
 u(2) = \rho(1)\rho(1), \\
 u(3) = \rho(1)\rho(2) + \rho(2)\rho(1), \\
 u(4) = \rho(1)\rho(3) + \rho(2)\rho(2) + \rho(3)\rho(1), \\
 u(5) = \rho(1)\rho(4) + \rho(2)\rho(3) + \rho(3)\rho(2) + \rho(4)\rho(1), \\
 ...... \\
 \end{array}
 \end{equation}
then we obtain negative value for $\rho(\frac{k}{R})$ due to the existence of the spike at $i=\frac{k}{R}$, which is not acceptable.
The DLT scheme in [8] is proposed to solve this problem. In DLT, $u(i)$ is divided into two parts, namely $u'(i)$ which contains degree $i=1, \frac{k}{R}$ and $u''(i)$ which contains the remaining degrees. The deconvolution is operated over $u''(i)$.  However it cannot provide the UEP for different sources. To address this problem, we apply our proposed degree distribution to  provide the UEP property for sources with different importance.

We propose distributed low-overhead LT (DLLT) codes by applying our degree distribution. Since the spike does not exist in the improved degree distribution, the deconvolution in (\ref{deconvolution}) can be directly operated. First the improved degree distribution \(u_{i}(i)\) is divided into two parts \(u_{i1}(i)\) and \(u_{i2}(i)\) as follows
\begin{equation}
u_{i1}(i) = \frac{{\rho_i (1) + \Gamma_i (1)}}{{\beta '}},
\end{equation}
\begin{equation}
u_{i2}(i) = \left\{ \begin{array}{l}
 \frac{{\rho_i (i) + \Gamma_i(i)}}{{\beta ''}}\;\;\;\frac{k}{R} - 1 \ge i \ge 2 \\
 \frac{{\rho_i (i)}}{{\beta ''}}\;\;\;\;\;\;\;\;\;\;\;\;\;k\; \ge i \ge \frac{k}{R},\\
 \end{array} \right.
\end{equation}
where \(\beta ' = \rho_i (1) + \Gamma_i (1)\) and \(\beta '' = \sum\limits_{i = 2}^k {u_{i2}(i)} \). Note $u_{i1}(i)$ only contains degree-1 nodes rather than the high degree-$\frac{k}{R}$ nodes which is hard to decode in low overhead regions. It simplifies the system providing the UEP, which will be discussed later.

As our degree distribution does not contain the spike, we could directly employ the deconvolution in the high degree (degrees larger than 1) part. The target degree distribution at source nodes \(p_{i}(i)\) is derived as:
\begin{equation}
p_{i}(i) = \sqrt {\frac{{\beta ''}}{\beta }} f_{i}(i) + (1 - \sqrt {\frac{{\beta ''}}{\beta }} )u_{i1}(i),
\end{equation}
where  \(f_{i}(i)*f_{i}(i) = u_{i2}(i)\).

The transmission scheme in DLLT is similar to that outlined in [8]. After receiving enough encoded symbols, the destination node can recover the initial source symbols by an LT decoding method, since the resulted codes are LT-like codes. Also since we use the improved degree distribution in DLLT codes, the BER performance of DLLT codes outperforms those of DLT codes in low overhead regions.

It is worth noting that based on the simplified structure of our proposed degree distribution, the DLLT transmission can be easily applied to an UEP-DLLT scheme (as follows). We assume symbols in one source node are more important than those in  another and they are referred to as more important bits (MIB) and less important bits (LIB) correspondingly. By slightly modifying the DLLT scheme above, the UEP property can be achieved.
We use a parameter \(s\) as the \emph{UEP factor}. Based on the DLLT transmission scheme above, \(p_{i}(i)\) is modified as:
\begin{equation}
p_{i}(i) = \left\{ \begin{array}{l}
 (\sqrt {\frac{{\beta ''}}{\beta }}  - s)f_i(i) + (1 - \sqrt {\frac{{\beta ''}}{\beta }}  + s){u_{i1}}(i) \;\;\;\;\;\;\;\\(source\:node\:of\:MIB),\\
 (\frac{{\frac{{\beta ''}}{\beta }}}{{\sqrt {\frac{{\beta ''}}{\beta }}  - s}})f_i(i) + (1 - \frac{{\frac{{\beta ''}}{\beta }}}{{\sqrt {\frac{{\beta ''}}{\beta }}  - s}}){u_{i1}}(i)\:  \;\;\;\;\;\;\;\;\;\;\;\\(source\:node\:of\:LIB).\\
 \end{array} \right.
\end{equation}

We denote \({p_1} = (1 - \sqrt {\frac{{\beta ''}}{\beta }}  + s) \) and \({p_2} = {(1 - \frac{{\frac{{\beta ''}}{\beta }}}{{\sqrt {\frac{{\beta ''}}{\beta }}  - s}})}\), where \({p_1}\) and \({p_2}\) are the selecting probability of \(u_{i1}(i)\) at MIB and LIB source node correspondingly. It is clear to see \({p_1} > {p_2}\). Thus, according to the DLLT transmission scheme, the encoded symbols will cover MIB with higher probabilities  while the source symbols in LIB are less likely to be covered by encoded symbols. More source symbols of MIB can be encoded into degree-1 encoded symbols at the receiver side, which in turn provide fast recovery of MIB when the receiver is in low overhead regions. Therefore MIB can achieve better protection, whereas LIB will have larger BER. Moreover, it is easy to see with \(s\) increasing, the gap in BER between MIB and LIB becomes greater, which will later be shown in the simulation results. We note that the UEP for distribtued LT codes is also studied in  \cite{TalariRa2012} mainly for error-floor regions. Yet our schemes are mainly for low overhead regions.

\section{Simulation results}

In this section, we present numerical results of our proposed schemes and compare them with some relevant results. The channels in all the simulations are assumed to be perfect, where the best performance can be achieved.

First, it is insightful to compare our results with those in [7], which is shown in Fig. 5. We note that the authors in [7] studied the data recovery in intermediate range. On the other hand, our work focuses on the low overhead regions.

For comparison, we select two degree distributions from Table.1 in [7], both designed for  $k=1000$. \(W_1 = 0.0624x + 0.5407{x^2} + 0.2232{x^4} + 0.1737{x^5}\) is designed for  $n=k$. \(W_2 = 0.1448x + 0.8552{x^2}\) is designed for $n=0.75k$. From the results it can be seen that since the work in [7] focuses on the intermediate range, the degree distribution $W_1$ and $W_2$ have better BER performance than our proposed degree distribution in intermediate range. In detail, $W_1$ has the best performance with $k=1000$ and  $W_2$ has the best performance with $k=750$. In low overhead regions ($k$ slightly larger than 1000), our proposed degree distribution has the best performance. Note RSD performs poorly in intermediate and low overhead regions.
\setcounter{figure}{4}
\begin{figure}[!ht] \centering
\includegraphics[scale=0.37]{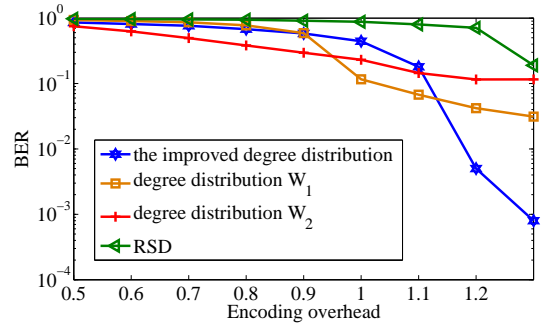}
\caption{ BER comparison of the three degree distributions, \( k=1000, c=0.15, \delta=0.2 \).}
\end{figure}

\begin{figure}[!ht] \centering
\includegraphics[scale=0.37]{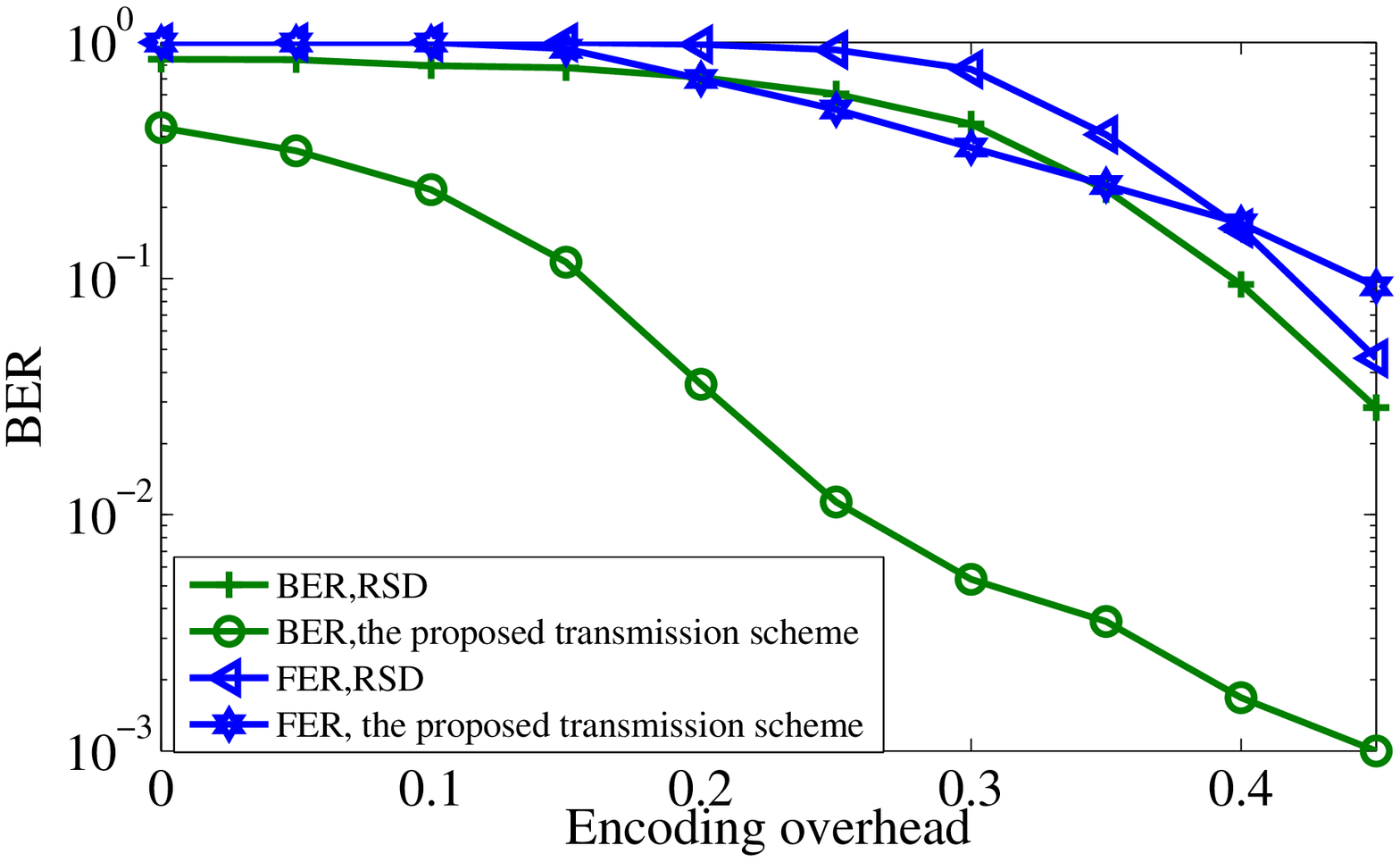}
\caption{BER and FER comparison of RSD and our proposed scheme, \( k=400, c=0.15, \delta=0.2 \)}
\end{figure}

\begin{figure}[!ht] \centering
\includegraphics[scale=0.37]{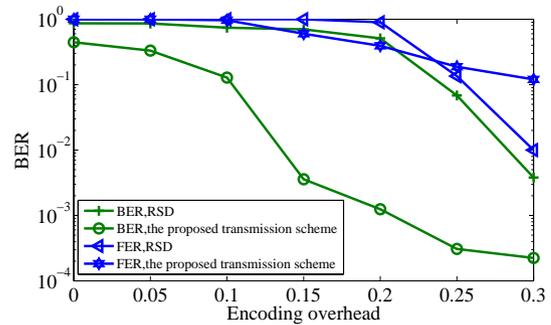}
\caption{BER and FER comparison of RSD and our proposed scheme, \( k=2000, c=0.15, \delta=0.2 \)}
\end{figure}
Fig. 6 and 7 show the BER and FER performance of the proposed transmission scheme for k=400 and 2000, respectively. The total number of encoded symbols is set to $n$, which means in the proposed transmission scheme, $n-2$ encoded symbols are generated by the LT encoder and two special encoded symbols are used for full recovery. It is demonstrated that the proposed transmission scheme based on our improved degree distribution has lower BER than RSD, especially in low overhead regions. Note when the overhead $\gamma$ is smaller than $\gamma_M$, our proposed codes outperform RSD in BER. Therefore the simulation results match well with the analysis in section IV.

Fig. 8 shows the UEP property of the proposed UEP-DLLT scheme. We can clearly see that MIB has better successful rate (1 - BER) than LIB especially with larger $UEP factor$, where successful rate means the number of the recovered source symbols normalized by the source block length.

\begin{figure}[!ht] \centering
\includegraphics[scale=0.37]{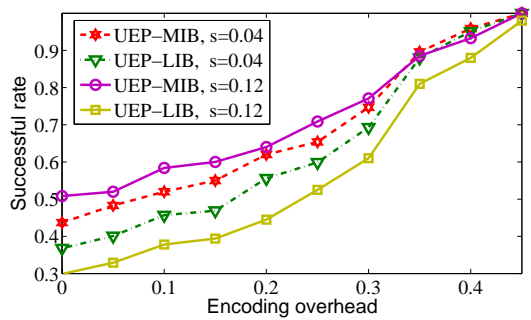}
\caption{ Successful rate comparison of UEP-MIB and UEP-LIB with difference $UEP factor$, \( k=400, c=0.15, \delta=0.2\) }
\end{figure}

\section{Conclusions}
\balance
We have proposed an improved degree distribution for LT codes which shows improved BER performance in low overhead regions and reduced complexity as compared to the conventional RSD. It is demonstrated through And-Or tree analysis that the proposed degree distribution outperforms RSD in BER performance with low overhead. To improve FERs in full recovery regions, we then proposed a transmission scheme based on the proposed degree distribution. Simulation results show the improved performance of the proposed transmission scheme. Finally, we extend our degree distribution to DLT codes for multi-source relay networks. It shows improved performance and meanwhile the UEP property is also achieved with our schemes.

\acks
This work was supported in part by China Mobile Research Institute and China National S\&T Major Project 2010ZX03003-003.

\end{document}